\documentclass{article}

\usepackage{arxiv}

\RequirePackage[colorlinks,citecolor=blue,linkcolor=blue,urlcolor=blue,pagebackref]{hyperref}
\usepackage[utf8]{inputenc} 
\usepackage[T1]{fontenc}    
\usepackage{hyperref}       
\usepackage{url}            
\usepackage{booktabs}       
\usepackage{amsfonts}       
\usepackage{nicefrac}       
\usepackage{microtype}      
\usepackage{xcolor}         
\usepackage{amsmath,amssymb}
\usepackage{natbib} 
\usepackage{comment}

\usepackage{amsthm}

\usepackage[ruled,vlined]{algorithm2e}
\usepackage{bm}

\usepackage{bbm}

\newcommand{\argmax}[1]{\underset{#1}{\text{argmax}}}

\newcommand{\eps}{\epsilon}
\newcommand{\dKL}{d_{\mathrm{KL}}}
\newcommand{\btheta}{\bm{\theta}}

\newcommand{\Beta}{\mathrm{Beta}}

\usepackage{amsthm}
\usepackage[capitalise]{cleveref}
\theoremstyle{plain}

\newtheorem{theorem}{Theorem}[section]
\newtheorem{lemma}[theorem]{Lemma}

\newtheorem*{theoremkasy}{Theorem~1 \citep{Kasy2021}}

\newtheorem{corollary}[theorem]{Corollary}
\newtheorem{proposition}[theorem]{Proposition}

\newtheorem{definition}[theorem]{Definition}

\newtheorem*{remark}{Remark}

\newcommand*{\email}[1]{\texttt{#1}}

\usepackage{authblk}

\newcommand\blfootnote[1]{%
  \begingroup
  \renewcommand\thefootnote{}\footnote{#1}%
  \addtocounter{footnote}{-1}%
  \endgroup
}

\title{Policy Choice and Best Arm Identification:\\
Asymptotic Analysis of Exploration Sampling}

\author[1,2]{Kaito Ariu}
\author[1]{Masahiro Kato\thanks{Corresponding author: \email{masahiro\_kato@cyberagent.co.jp}}$\ \ $}
\author[3]{Junpei Komiyama}
\author[4]{Kenichiro McAlinn}
\author[5]{Chao Qin}

\affil[1]{AI Lab, CyberAgent, Inc.}
\affil[2]{School of Electrical Engineering and Computer Science, KTH}
\affil[3]{Stern School of Business, New York University} 
\affil[4]{Fox School of Business, Temple University}
\affil[5]{Columbia Business School, Columbia University}

\begin{document}

\maketitle

\begin{abstract}
\blfootnote{We thank Edoardo Airoldi, Hidehiko Ichimura, Daniel Russo, and Xuedong Shang for insightful comments and discussion. This draft was produced following communication with the authors of the original manuscript.}
We consider the ``policy choice'' problem-- otherwise known as best arm identification in the bandit literature-- proposed by \citet{Kasy2021} for adaptive experimental design. 
Theorem~1 of \citet{Kasy2021} provides three asymptotic results that give theoretical guarantees for exploration sampling developed for this setting. We first show that the proof of Theorem~1~(1) has technical issues, and the proof and statement of Theorem~1~(2) are incorrect.
We then show, through a counterexample, that Theorem~1~(3) is false.
For the former two, we correct the statements and provide rigorous proofs. For Theorem~1~(3), we propose an alternative objective function, which we call posterior weighted policy regret, and derive the asymptotic optimality of exploration sampling. 
\end{abstract}



\section{Introduction}
\citet{Kasy2021} proposes what the authors call the ``policy choice'' problem for adaptive treatment assignment in experiments.
The goal in policy choice is to choose a policy that is the {\it best} treatment amongst a set of treatments within several waves of an experiment.  
To evaluate algorithms in this setting, the authors propose a metric called ``policy regret.''
Using this metric, they develop a dynamic programming algorithm to optimize the expected policy regret (expected social welfare), but find that the proposed algorithm is computationally intractable. 
In light of this, the authors propose an algorithm called ``exploration sampling'' and prove its asymptotic optimality (Theorem~1).

The purpose of this paper is to show, first, that Theorem~1 is incorrect.
In particular, we show that the proof of Theorem~1~(1) has technical issues, proof and statement of Theorem~1~(2) is incorrect, and Theorem~1~(3) is false, which we show through a counterexample.
We then provide a corrected version of Theorem~1, as well as the associated corrected lemmata, with rigorous proofs.
As Theorem~1~(3) is false under the expected policy regret, we propose the posterior weighted policy regret, which we then show the asymptotic optimality of exploration sampling under this objective.
We further extend the theoretical results by relaxing the assumptions on the prior specification.
These results provide theoretical support and extend the applicability of exploration sampling in adaptive experiments.

This paper is organized as follows.
We review the problem setting and main theoretical results in \citet{Kasy2021} in Section~\ref{sec:problem}.
The problems with the main theorem are discussed in Section~\ref{sec:incorrect}.
In Section~\ref{sec:our_correct}, we provide our corrected theorem (Theorem~\ref{thm:correct_kasy}).
The proof of Theorem~\ref{thm:correct_kasy} and the corrected lemmata used in our proof is in Appendix~\ref{sec:correction} and Appendix~\ref{sec:corrected_lemma}, respectively.

\section{Problem Setting and Main Results in Kasy and Sautmann (2021)}
\label{sec:problem}
We first review the problem setting and theoretical result of \citet{Kasy2021}.

\subsection*{Problem Setting}
Suppose that there are multiple treatments, which are also called ``arms'' in the multi-armed bandit literature. Let $k\geq 2$ be the number of possible treatments. In the policy choice problem, a policymaker is interested in the expected outcome of the treatments. The outcome corresponding to each treatment is a binary random variable. For estimating the expected values, the policymaker conducts adaptive experiments with multiple waves. At the end of each wave, the policymaker observes the outcomes and updates treatment assignment in subsequent waves based on past observations. Upon completion of the waves, the policymaker chooses the treatment that yields the highest expected outcome, which is called policy. The goal of policy choice \citep{Kasy2021}, also known as best arm identification\footnote{Here, "identification" simply means selecting the arm (treatment) with the highest expected reward, and not the identification problem in the econometric literature.} (BAI), is to propose an optimal experimental design for this setting.

The experiment consists of waves $t=1,\dots,T$, where at each wave $t$, there is a new random sample of $N_t$ experimental units, $i = 1, \dots, N_t$, drawn from the population of interest. We denote the total sample size by $M=\sum^T_{t=1}N_t$. For each unit $i$ in period $t$, the experimenter can assign one of $k$ different treatments $D_{i,t}\in\{1,\dots, k\}$ and then observe a binary outcome $Y_{i,t}\in\{0, 1\} = \sum^k_{d=1}\mathbbm{1}\{D_{i,t} = d\}Y^d_{i,t}$, where the potential outcome vectors $(Y^1_{i,t},\dots, Y^k_{i,t})$ for unit $i$ in period $t$ are i.i.d. draws. Let us assume that each treatment $d\in\{1,\ldots,k\}$ has the average potential outcome and denote it as $\theta^d = \mathbb{E}\left[Y^d_{i,t}\right]$.
Under this setting, the goal is to choose the policy with the highest expected outcome amongst the given multiple treatments. \citet{Kasy2021} denotes the true optimal treatment by $d^{(1)} \in \argmax{d\in\{1,\ldots,k\}}\,\theta^{d}$, and let $\Delta^d = \theta^{d^{(1)}} - \theta^d$ be the policy regret when choosing treatment $d\in\{1,\ldots,k\}$, relative to the optimal treatment $d^{(1)}$. This performance metric is refereed to as the simple regret in the BAI literature \citep{Audibert2010,lattimore2020}.

At the beginning of wave $t$, there are $N_t$ units available, and the experimenter can optimize the allocation of the treatments to these $N_t$ units.
In \citet{Kasy2021}, 
the treatment assignment
in wave $t$ is summarized by the vector $\bm n_t = (n^1_t, \dots, n^k_t)$ with $\sum_{d=1}^k n^d_t = N_t$.  For each treatment $d$, denote the number of successes among $n^d_t$ units in wave $t$ by $s^d_t = \sum_{i=1}^{N_t}\mathbbm{1}\{D_{i,t} = d, Y_{i,t} =1\}$. Let us denote the outcome of wave $t$ by the vector $\bm s_t = (s^1_t, \dots, s^k_t)$, where $s^d_t \leq n^d_t$, which can be observed at the end of wave $t$. We denote the cumulative versions of these terms from $1$ to $t$ by $m^d_t = \sum_{t'\leq t} n^d_{t'}$, $r^d_t = \sum_{t' \leq t} s^d_{t'}$, and $\bm m_t = (m^1_t,\dots, m^k_{t})$, $\bm r_t = (r^1_{t},\dots, r^k_{t})$. 

The policymaker holds prior belief $\mathrm{Beta}(\alpha^d_0, \beta^d_0)$ for treatment $d\in\{1,\ldots,k\}$.
In \citet{Kasy2021}, the uniform prior is used as the default for applications, i.e., $\alpha^d_{0} = \beta^d_0 = 1$ for all $d$. The posterior belief is defined by the parameters $(\alpha^d_t, \beta^d_t) = (\alpha^d_0 + r^d_{t-1}, \beta^d_0 + m^d_{t-1} - r^d_{t-1})$. 

\citet{Kasy2021} gives per-capita expected social welfare of policy $d$ as 
\[\mathrm{SW}_T(d) = \mathbb{E}\left[\theta^d| \bm m_T, \bm r_T\right] = \frac{\alpha^d_0 + r^d_T}{\alpha^d_0 + \beta^d_0 + m^d_T},\]
and proposes choosing a policy as $d^*_T \in \argmax{d\in\{1,\ldots,k\}}\,\mathrm{SW}_T(d)$. 

\subsection*{Expected Policy Regret} 
In \citet{Kasy2021}, the treatment assignment algorithms are evaluated by the expected social welfare, or, equivalently, expected policy regret, which is defined for the policy $d^*_T$ as follows: 
\begin{align}\label{ineq_bairegret}
\mathrm{R}_{\bm{\theta}}(T) = \mathbb{E}\left[\Delta^{d^*_T}| \bm \theta\right] = \sum_{d=1}^k \Delta^d \cdot \mathbb{P}(d^*_T = d| \bm \theta),
\end{align}
where $T$ is the number of experimental waves, and the expectation is taken over all possible successes and assignment choices for treatments.
This objective is identical to the expected simple regret in the BAI literature.

\subsection*{Exploration Sampling}  
In each wave $t$, we define the posterior probability that the treatment $d\in\{1,\ldots,k\}$ is the optimal treatment as 
\[p^d_t = \mathbb{P}\left(d = \argmax{d'\in\{1,\ldots,k\}}\, \tilde{\theta}^{d'}\mid \bm m_{t- 1}, \bm r_{t-1}\right),\]
where $\tilde{\bm{\theta}} = (\tilde{\theta}^1,\ldots,\tilde{\theta}^k)$ is a sample drawn from the posterior.
Thompson sampling can be interpreted as a method that assigns $\lfloor p^d_t N_t\rfloor$ observations to treatment $d$. Based on this idea and the result of \citet{Russo2016}, \citet{Kasy2021} proposes their treatment assignment algorithm called exploration sampling. In exploration sampling, we assign $\lfloor q^d_t N_t\rfloor$ of observations to treatment $d$, where 
\begin{align*}
    q^d_t = S_t\cdot p^d_t \cdot (1-p^d_t),
\end{align*}
with the normalization term $S_t = \left(\sum_{d=1}^k p^d_{t}\cdot (1- p^d_t)\right)^{-1}$.
\citet{Kasy2021} analyzes the theoretical properties of this algorithm and provides asymptotic guarantees.

\begin{remark}[BAI]
This problem setting is known as BAI in the multi-armed bandit literature \citep{even2002pac, Mannor2004, EvanDar2006, Audibert2010}. 
Though the problem of BAI itself goes back decades, variants go as far back as the 1950s, in the context of sequential testing problems \citep{Wald1945,Chernoff1959}. Some of the earliest advances on this topic are summarized in \citet{bechhofer1968sequential}. Another literature on ordinal optimization has been studied in the operation research community and a modern formulation was established in the 2000s \citep{chen2000,glynn2004large}. Most of those studies have considered the estimation of optimal allocations separately from the error rate under known optimal allocations. In the 2010s, the machine learning community reformulated the problem to synthesize both issues and explicitly discussed them.
For a more detailed survey, see, for example, the Introduction section of \citet{Kaufman2016complexity} and Section~33 of \citet{lattimore2020}. 
\end{remark}

\subsection*{Main Theorem in \cite{Kasy2021}}
\cite{Kasy2021} provides the following performance guarantees for exploration sampling.

\begin{theoremkasy}
Consider exploration sampling, with fixed wave
size $N_t = N \geq 1$. 
Assume that the optimal arm $d^{(1)} = \underset{d\in\{1,\ldots,k\}}{\rm{argmax}}\,\theta^d$ is unique and that $\theta^{d^{(1)}} < 1$. As $T\to \infty$, the followings hold:
\begin{description}
    \item[(1)] The share of observations $\frac{m^{d^{(1)}}_T}{NT}$ assigned to the best treatment $d^{(1)}$ converges in probability to $\frac{1}{2}$, that is, 
    \[
    \frac{m^{d^{(1)}}_T}{NT} \xrightarrow{p}  \frac{1}{2}.
    \]
    \item[(2)] The share of observations $\frac{m^d_T}{NT}$ assigned to each treatment $d\neq d^{(1)}$ converges in probability to a non-random share $\rho^d$, that is, 
    \[
    \frac{m^d_T}{NT} \xrightarrow{p} \rho^d, \quad \forall d\neq d^{(1)},
    \]
    where the limit assignment shares $\bm{\rho} = (\rho^1, \ldots, \rho^k)$ (with $\rho^{d^{(1)}} =\frac{1}{2}$) is such that $-\frac{1}{NT}\log p^d_T \xrightarrow{p} \Gamma^*$ for some $\Gamma^* > 0$ that is constant across $d\neq d^{(1)}$.
    \item[(3)] Expected policy regret converges to $0$ at the same rate $\Gamma^*$, that is, 
    \[
    -\frac{1}{NT}\log \mathrm{R}_{\bm{\theta}}(T) \xrightarrow{} \Gamma^*.
    \]
    No algorithm with limit assignment shares $\hat{\bm{\rho}}\neq \bm \rho$ with $\hat{\rho}^{d^{(1)}} = \frac{1}{2}$ exists for which $\mathrm{R}_{\btheta}(T)$ goes to $0$ at a faster rate than $\Gamma^*$.
\end{description}

\end{theoremkasy}

While the introduction of the policy choice setting into the field of economics is laudable, as well as conducting a field experiment that applies this methodology to actual policy experiments, unfortunately, this theorem has several issues, which we will expound and correct below.

\section{Incorrectness of Theorem~1 in Kasy and Saumtmann (2021)}
First, we describe the incorrectness of Theorem~1~(1) and (2). Then, we provide a counterexample to Theorem~1~(3).
\label{sec:incorrect}
\subsection*{Incorrect Proof for Theorem~1~(1) and (2)}
In their proof, \citet{Kasy2021} refers to \citet{Russo2016} when showing posterior convergence. However, the results of  \citet{Russo2016} do not guarantee the performance of their algorithm with the Beta--Bernoulli model. This is because Assumption~1 of \citet{Russo2016} requires the boundedness on the first derivative of the log-partition function of the reward distribution belonging to the exponential family. This assumption is violated when the parameter space of the Bernoulli models is $[0, 1]$. Therefore, when using a beta prior whose support covers $[0,1]$, we cannot apply the results of \citet{Russo2016}. For a more detailed discussion, see Section~5 of \citet{Shang2020}\footnote{In Section~5 of \citet{Shang2020}, the authors explain that ``\citet{Russo2016} proves a similar theorem under three confining boundedness assumptions (see \citet{Russo2016}, Assumption 1) on the parameter space, the prior density and the (first derivative of the) log-normalizer of the exponential family. Hence, the theorems in \citet{Russo2016} do not apply to the two bandit models most used in practice and considered in this paper: the Gaussian and Bernoulli model. In the first case, the parameter space is unbounded; in the latter, the derivative of the log-normalizer
(which is $e^\eta/(1+e^\eta)$) is unbounded.''}.
This problem has already been pointed out by \citet{Russo2016} and \citet{Shang2020}. To show the posterior convergence rate for the Beta--Bernoulli model with the general $\mathrm{Beta}(\alpha_0^d, \beta_0^d)$ prior, a separate proof is needed for each algorithm. For instance, \citet{Shang2020} proves posterior consistency when considering the top-two Thompson sampling (TTTS). This proof is non-trivial, which \citet{Kasy2021} does not show.

Additionally, the limit of the convergence is different from that of \citet{Russo2016}. In \citet{Kasy2021}, $\Gamma^*$ and $\bm{\rho} = (\rho^1, \ldots, \rho^k)$ are given as the optimal value and solutions of the following optimization problem:
\begin{align}
\label{ineq_glynnopt}
&\max_{\bm{w}\in\mathbb{R}^k}\ \Gamma\\
&\mathrm{s.t.}\ G_d\left(w^{d^{(1)}}, w^d\right) - \Gamma \ge 0, \quad \forall d \neq d^{(1)},\nonumber\\
&\ \ \ \ \ \sum_{d=1}^k w^d = 1,\quad w^{d^{(1)}} = \frac{1}{2},\quad w^d \ge 0, \quad\forall d \neq d^{(1)},\nonumber
\end{align}
where
\begin{equation*}
G_d \left(w^{d^{(1)}}, w^d\right) = \min_{x \in \left[\theta^d, \theta^{d^{(1)}}\right]}  \left[w^{d^{(1)}} \dKL\left(x, \theta^{d^{(1)}}\right) + w^d \dKL\left(x, \theta^d\right)\right],
\end{equation*}
and $\dKL(p, q) := p \log(p/q) + (1-p)\log((1-p)/(1-q))$ is the KL divergence between two Bernoulli distributions.
However, these terms are different from that introduced in \citet{Russo2016}, which is defined as the value $\Lambda^*$ and the solution $\bm{\lambda} = (\lambda^1, \ldots, \lambda^k)$ of the following optimization problem:
    \begin{align}\label{ineq_shangopt}
&\max_{\bm{w}\in\mathbb{R}^k}\ \Lambda\\
&\mathrm{s.t.}\ C_d\left(w^{d^{(1)}}, w^d\right) - \Lambda \ge 0, \quad \forall d \neq d^{(1)},\nonumber\\
&\ \ \ \ \ \sum_{d=1}^k w^d = 1,\quad w^{d^{(1)}} = \frac{1}{2},\quad w^d \ge 0, \quad\forall d \neq d^{(1)},\nonumber
\end{align}
where
\begin{equation*}
C_d \left(w^{d^{(1)}}, w^d\right) = \min_{x \in \left[\theta^d, \theta^{d^{(1)}}\right]}  \left[w^{d^{(1)}} \dKL\left(\theta^{d^{(1)}}, x\right) + w^d \dKL\left(\theta^d, x\right)\right]. 
\end{equation*}
The major difference between these optimization problems is that the arguments of the Kullback–Leibler divergence are reversed.\footnote{Sandeep Juneja and Daniel Russo have also discovered this asymmetry. Chao Qin thanks them for the discussion.}
Intuitively speaking, the value $\Gamma^*$ of \citet{glynn2004large} characterizes the likelihood of incorrectly recommending treatment when the optimal treatment is $d^{(1)}$, whereas the value $\Lambda^*$ of \citet{Russo2016} characterizes the likelihood of the observed data under the hypothesis that the optimal treatment is \textit{not} treatment $d^{(1)}$. 
The latter notion is appropriate here as we are minimizing the mass of the posterior where the optimal treatment is different.

\begin{remark}[Almost sure convergence and convergence in probability]\label{rem_as}
\rm
The original statement of \citet{Russo2016} states $-\frac{1}{NT}\log \left(\sum_{d\neq d^{(1)}}p^d_T\right) \xrightarrow{\mathrm{a.s.}} \Lambda^*$, whereas \citet{Kasy2021} incorrectly cites the result as $-\frac{1}{NT}\log p^d_T \xrightarrow{p} \Gamma^*$.
Almost sure convergence describes an event that holds on each sample path, whereas convergence in probability describes an event that holds for a single $t$, which is weaker than almost sure convergence.

During the proofs, one can find many steps that essentially require pathwise discussions. 
For example, to apply the law of large numbers, it is required that we have infinitely many samples for almost all sample paths.
The derivation of the limit rate $\Lambda^*$ in our corrected theorem is via a bound of the posterior probabilities of the form $\exp(-(1\pm\epsilon)NT\Lambda^*)$ that holds over all $t > T_0(\eps)$ in Appendix~\ref{appdx:lem:kasy22} (proof of Lemma~\ref{lem:kasy22}).
\end{remark}

\begin{remark}[General prior]
\rm
Results in \citet{Shang2020} are limited to the uniform prior (i.e., $\Beta(1,1)$). In this paper, we extend the results in \citet{Shang2020} to $\Beta(\alpha_0^d,\beta_0^d)$ priors for each $d \in \{1, \ldots, k\}$, where the constants $\alpha_0^1, \ldots, \alpha_0^k, \beta_0^1, \ldots ,\beta_0^k>0$ can be arbitrary. 
\end{remark}

\subsection*{Comments on Theorem~1~(3)}
There are two ways in which Theorem~1~(3) is problematic. 
First, we show that the proof is incorrect. 
Second, we show that the statement contradicts an existing theoretical result \citep{Carpentier2016}.

\subsubsection*{Incorrect Proof of Theorem~1~(3)}\label{subsec_incorrectproof}
Theorem~1~(3) cannot be a consequence of Theorem~1~(1) and (2). Theorem~1~(3) quantifies the asymptotic convergence rate of the expected policy regret, $ \mathrm{R}_{\bm{\theta}}(T)$, which is equivalent to the weighted sum of the asymptotic probability of misidentification, $\sum_{d=1}^k\Delta^d\mathbb{P}\left(d^*_T = d| \bm{\theta}\right)$. In order to evaluate this, one needs to quantify the convergence rate of $-\frac{1}{NT}\log p^d_t$ to the optimal treatment allocation $\Gamma^*$. However, Theorem~1~(1) and (2) only state the consistency of the optimal treatment allocation, without providing the convergence rate. 

A convergence in probability states that, for any $\eps>0$ and $\delta > 0$, there exists $t_0(\epsilon, \delta)$, such that for all $t> t_0(\epsilon, \delta)$,
\begin{equation}
\label{eq:conv_prob}
\mathbb{P}\left(\left|-\frac{1}{NT}\log p^d_T - \Gamma^*\right| > \eps\right) \le \delta.
\end{equation}
From the convergence in probability, we can show that the expected policy regret converges to $0$ as $T\to \infty$ ($\mathrm{R}_{\btheta} (T) = o(1)$), but we cannot derive Theorem~1~(3), i.e., $-\frac{1}{NT}\log \mathrm{R}_{\btheta}(T) \to \Gamma^*$ as $T \to \infty$, because there exist counterexamples where the speed of convergence can be insufficient. 

For example, convergence in probability means that it can include the following examples of convergence: a relationship $t_0 = \max(1/\epsilon^2, 1/\delta^2)$. In this case, for any $\varepsilon$ and for all $t \ge t_0(\epsilon, \delta)$, the inequality \eqref{eq:conv_prob} holds for $ \delta \ge 1/\sqrt{t}$. In this case, the probability of convergence cannot be guaranteed to be greater than $1-1/\sqrt{t}$. It leads to a polynomial order  $1/\sqrt{t}$ of the expected policy regret. Thus, convergence in probability does not lead to the expected policy regret of the same rate.

\subsubsection*{Counterexample to Theorem~1~(3)}\label{subsec_carpentier}
The result derived in \cite{Carpentier2016} contradicts Theorem~1~(3). 
By utilizing information theoretic arguments, \cite{Carpentier2016} constructs a counterexample where the expected policy regret can be strictly larger than what is expected from the results\footnote{The difference between $\Gamma^*$ and $\Lambda^*$ does not matter in this counterexample; Theorem \ref{thm_lower} utilizes the Pinsker's inequality, which bounds both $\Gamma^*$ and $\Lambda^*$ from below.} of \citet{glynn2004large}. 

The crux here is that the optimal allocation computed from $\Gamma^* = \Gamma^*(\btheta)$ depends on the true parameter $\btheta$. Let us hypothetically consider several policy choice problems under different parameters. We call each policy choice problem a problem instance, which is characterized by the parameter $\btheta$. For an easy problem instance (i.e., $\btheta$ such that $\Gamma^*(\btheta)$ is large), the expected policy regret must decay faster, whereas for a hard problem instance (i.e., $\btheta$ such that $\Gamma^*(\btheta)$ is small) the expected policy regret decays slower. Suppose that there exist $k$ problem instances under $k$ different parameters $\btheta_1, \dots, \btheta_k$, which are indexed by $1,\dots,k$, respectively. Here, $k$ is the same as the number of treatments in \citet{Kasy2021}. \cite{Carpentier2016} constructs a particular set of $k$ problem instances such that the expected policy regret of an algorithm converges at the rate slower than $\Gamma^*(\btheta)$ for at least one of the $k$ problem instances.

In the following, we assume unit wave size $N = 1$ (and thus $M=NT=T$), which is usually adopted in the BAI literature:
\begin{theorem}{\rm (Lower bound on the expected policy regret.)}\label{thm_lower}
There exists a problem instance $\btheta$ and an infinite subsequence of integers $\{T_n\}_{n=1}^\infty$ such that for any $T_n$, the expected policy regret of any algorithm is lower bounded as
\begin{equation*} \label{ineq_cgammarate}
\mathrm{R}_{\btheta}(T_n) \ge \exp\left(-\frac{C}{\log(k)}\Gamma^*(\btheta)T_n\right) 
\end{equation*}
for some constant $C>0$.
\end{theorem}
We present the proof of Theorem~\ref{thm_lower} in Appendix~\ref{app:carpentier}. To make the connection with the result of \citet{Kasy2021} clearer, we transform the inequality in Theorem~\ref{thm_lower} as for any $T_n$,
\begin{equation} \label{ineq_cgammarate_kasynot}
-\frac{1}{T_n}\log  \mathrm{R}_{\btheta}(T_n) \le \frac{C}{\log(k)} \Gamma^*(\btheta).
\end{equation}
Since $\log(k) \rightarrow \infty$ as $k \rightarrow \infty$, the expected policy regret decays arbitrarily slower than what Theorem~1~(3) claims.

\section{Corrected Main Theorem}
\label{sec:our_correct}
In this paper, we provide a correction to all three statements in Theorem~1. Since the asymptotic optimality on the expected policy regret, used in original Theorem~1~(3), is unattainable, we propose a new objective-- posterior weighted policy regret-- and derive the asymptotic optimality under this objective. We introduce the following objective
\begin{align*}
{\mathrm{W}}_{\btheta}(T) &= \sum_{d\in\{1,\dots,k\}}\Delta^{d}\cdot p^d_T \\
&=  \sum_{d\in\{1,\dots,k\}}\Delta^{d} \cdot \mathbb{P}\left(d = \argmax{d'\in\{1,\ldots,k\}}\, \tilde{\theta}^{d'} \mid \bm m_{T-1}, \bm r_{T-1}\right),
\end{align*}
where the gap $\Delta^d = \theta^{d^{(1)}} - \theta^d$ depends on the unknown true parameter $\btheta = (\theta^1,\ldots,\theta^k)$, while the posterior probability above is based on the random sample $\tilde{\bm{\theta}} = (\tilde{\theta}^1,\ldots,\tilde{\theta}^k)$ drawn from the posterior. By comparing the posterior probability of making the wrong decision $\sum_{d\neq d^{(1)}} p^d_T$, this objective takes into account the magnitude between the best policy and sub-optimal policies.

Our objective is a synthesis of the frequentist objective (policy regret) and Bayesian objective {(expected posterior policy regret)}, as it combines the regret for not choosing the ``true'' policy (frequentist), which is unobserved by the policymaker, and the magnitude reflected in the posterior distribution (Bayesian), which the policy maker updates.
Unlike the expected policy regret, defined in \citet{Kasy2021}, we can derive the optimality of the exploration sampling in view of this objective.

\begin{theorem}[Corrected Theorem~1 of \citet{Kasy2021}]
\label{thm:correct_kasy}
Consider exploration sampling, for Bernoulli bandits with $\Beta(\alpha_0^d,\beta_0^d)$ priors for each $d \in \{1, \ldots, k\}$, where the constants $\alpha_0^1, \ldots, \alpha_0^k, \beta_0^1, \ldots ,\beta_0^k>0$ can be arbitrary. Let the wave
size be fixed $N_t = N \geq 1$. 
Assume that the optimal arm $d^{(1)} = \underset{d\in\{1,\ldots,k\}}{\rm{argmax}}\,\theta^d$ is unique and that $\theta^{d^{(1)}} < 1$. The following statements hold:
\begin{description}
    \item[(1)] The share of observations $\frac{m^{d^{(1)}}_T}{NT}$ assigned to the best treatment $d^{(1)}$ converges almost surely to $\frac{1}{2}$, that is, 
    \[
    \frac{m^{d^{(1)}}_T}{NT} \xrightarrow{\mathrm{a.s.}} \frac{1}{2}.
    \]
    \item[(2)]
    The share of observations $\frac{m^{d}_T}{NT}$ assigned to each treatment $d\neq d^{(1)}$ converges almost surely to a non-random share $\lambda^d$, that is, 
    \[
    \frac{m^d_T}{NT} \xrightarrow{\mathrm{a.s.}} \lambda^d, \quad \forall d\neq d^{(1)},
    \]
    where the limit assignment shares $\bm{\lambda} = (\lambda^1, \ldots, \lambda^k)$ (with $\lambda^{d^{(1)}} = \frac{1}{2}$) is such that $-\frac{1}{NT}\log\left( \sum_{d\neq d^{(1)}} p^d_T\right) \xrightarrow{\mathrm{a.s.}} \Lambda^*$, which is the optimal value of the optimization problem defined in \eqref{ineq_shangopt}.
    \item[(3)] The posterior weighted policy regret converges almost surely to $0$ at the same rate $\Lambda^*$, that is,
    \[-\frac{1}{NT}\log \mathrm{W}_{\btheta}(T) \xrightarrow{\mathrm{a.s.}}  \Lambda^*.\]
    No algorithm with limit assignment shares  
    $\hat{\bm{\lambda}}\neq \bm{\lambda}$ with $\hat{\lambda}^{d^{(1)}} = \frac{1}{2}$
    exists for which $\mathrm{W}_{\btheta}(T)$ goes to $0$ at a faster rate than $\Lambda^*$.
\end{description}
\end{theorem}

To prove Theorem~\ref{thm:correct_kasy}, we require Lemmata~2 (2), 4--6 in \citet{Kasy2021} (Lemmata~1, 2 (1), and 3 are irrelevant for the corrected theorem). 
The original version of these lemmata has technical issues that require correction. Lemma~2 (2) in \citet{Kasy2021} cites \citet{Russo2016}, which cannot be applied to Beta-Bernoulli models with  $\mathrm{Beta}(\alpha^d_0,\beta^d_0)$ priors.
Lemma~4 cites \citet{Russo2016}, though, the statement is adapted from the original: while Lemma~4 supposes that $m^{d^{(1)}}_T/(NT)\xrightarrow{\mathrm{p}} \beta$ for a constant $\beta > 0$, \citet{Russo2016} supposes $\lim_{T\to\infty}m^{d^{(1)}}_T/(NT) =\beta$.
Moreover, Lemmata~2 (2), 4--6 are incorrectly adapted from \cite{Russo2016} regarding the convergence of random variables.
To address these issues, we show the corrected versions of Lemmata~2 (2), 4--6 in Appendix~\ref{sec:corrected_lemma}.


\begin{appendix}

\section*{}
\end{appendix}

\begin{appendix}

\section{Proof of Theorem~\ref{thm_lower}}\label{app:carpentier}
Let us define the problem complexity as
\begin{equation*} 
H(\btheta) = \sum_{d \notin {\argmax{d' \in \{1, \ldots, k\}} \theta^{d'}}} \frac{1}{\left(\Delta^d\right)^2} 
= \sum_{d \notin {\argmax{d' \in \{1, \ldots, k\}} \theta^{d'}}} \frac{1}{\left(\theta^{d^{(1)}} - \theta^d\right)^2},
\end{equation*}
where $d^{(1)} \in \underset{d'\in\{1,\ldots,k\}}{\rm{argmax}}\ \theta^{d'}$.
To show Theorem~\ref{thm_lower}, we use the following results. 

\begin{lemma}{\rm (Lower bound on $\Gamma^*(\btheta)$)}\label{lem_gamma}
Consider $\btheta =(\theta^1,\ldots,\theta^k)$ such that $d^{(1)} = \underset{d\in\{1,\ldots,k\}}{\rm{argmax}}\ \theta^d$ is unique. We have
\begin{equation*}
\Gamma^*(\btheta) \ge \frac{1}{2H(\btheta)}.
\end{equation*}
\end{lemma}

\begin{proof}
Recall that $\Gamma^* (\btheta)$ is the solution of the optimization problem defined in \eqref{ineq_glynnopt}. We have,
\begin{align*}
\Gamma^*(\btheta)
&= \max_{\bm{w}}\min_{d\neq d^{(1)}} G_j \left(w^{d^{(1)}}, w^d \right) 
\\
&\ge 2 \max_{\bm{w}}\min_{d\neq d^{(1)}} \min_{x \in \left[\theta^d, \theta^{d^{(1)}}\right]}  
\left[w^{d^{(1)}} \left(x - \theta^{d^{(1)}}\right)^2 + w^d \left(x - \theta^d\right)^2\right]\\
&\quad\quad \text{\ \ \ (by Pinsker's inequality)}
\\
&\ge 2 \max_{\bm{w}}\min_{d\neq d^{(1)}} \min_{x \in \left[\theta^d, \theta^{d^{(1)}}\right]}  
\min\{w^{d^{(1)}}, w^d\} \cdot \left[\left(x - \theta^{d^{(1)}}\right)^2 + \left(x - \theta^d\right)^2 \right]\\
&= 2 \max_{\sum_{d\neq d^{(1)}}w^d=1/2}\min_{d\neq d^{(1)}} \min_{x \in \left[\theta^d, \theta^{d^{(1)}}\right]}  
w^d  \left[\left(x - \theta^{d^{(1)}}\right)^2 + \left(x - \theta^d\right)^2 \right]\\
&\quad\quad\ \ \ \left(\text{since }w^{d^{(1)}} = 1/2\geq w^d \right)\\
&\ge  \max_{\sum_{d\neq d^{(1)}}w^d=1/2}\min_{d\neq d^{(1)}} w^{d} \left(\theta^{d^{(1)}} - \theta^d\right)^2\\
&\quad\quad\ \ \ \left(\text{since the inner minimization achieves optimality at }x=\frac{\theta^{d^{(1)}}+\theta^d}{2}\right)\\
&\geq \frac{1}{2}\left(\sum_{d \ne d^{(1)}}\frac{1}{\left(\theta^{d^{(1)}} - \theta^d\right)^2}\right)^{-1}   
= \frac{1}{2H(\btheta)}
\end{align*}
where the last inequality holds by considering a specific choice $\left(w^d\right)_{d\neq d^{(1)}}$ with $\sum_{d\neq d^{(1)}}w^d=1/2$, such that
\[
w^d = \frac{1}{2}\frac{1}{\left(\theta^{d^{(1)}} - \theta^d\right)^2}\left(\sum_{d' \ne d^{(1)}}\frac{1}{\left(\theta^{d^{(1)}} - \theta^{d'}\right)^2}\right)^{-1}, \quad \forall d\neq d^{(1)}.
\]
\end{proof}

Next, we define $k$ problem instances as follows. Let us denote the parameter of the problem instance $i\in\{1,\dots, k\}$ by $\btheta_{i} = (\theta^1_{i}, \dots, \theta^k_{i})$, i.e., the instance $i$'s average potential outcome of the treatment $d$ is denoted as $\theta^d_i$. 
\begin{definition}{\rm (Problem instances)}\label{def_counterexample}
The first problem instance denoted by $\btheta_{1} = \left(\theta_1^1,\ldots,\theta_1^k\right)$ is defined as
\begin{align*}
\theta^{1}_{1} &= \frac{1}{2} \quad \text{and} \quad \theta^d_{1} = \frac{1}{2} - f^d \quad \forall d\neq 1,
\end{align*}
where $f^d := \frac{d}{4k}$. \\
The other $(k-1)$ problem instances are denoted by $\btheta_{i}$ for $i = 2,3,\dots,k$, and each $\btheta_{i}=(\theta_i^1,\ldots,\theta_i^k)$ is defined as the ones where the average potential outcome of the treatment $i$ is replaced by $\frac{1}{2} + f^i$ so that treatment $i$ is the best treatment in the problem instance $i$, that is,
\begin{align*}
\theta^i_{i} = \frac{1}{2} + f^i \quad \text{and} \quad \theta^{d}_{i} &= \theta^{d}_{1} \quad \forall d\neq i.
\end{align*}
\end{definition}
We use the following proposition.
\begin{proposition}{\rm (Theorem~2 in \cite{Carpentier2016}.)}\label{lem_lower_inner}
Let $k \geq 2$. Consider the problem instances $\{\btheta_{1},\dots,\btheta_{k}\}$ in Definition \ref{def_counterexample}. For any algorithm, for each $T$, there exists at least one $\btheta \in \{\btheta_{1},\dots,\btheta_{k}\}$ such that 
\begin{equation*}
\mathrm{R}_{\btheta} (T) \ge \frac{1}{6} \exp\left(- \frac{60T}{h^* H\left(\btheta\right)} - 2\sqrt{T \log(6Tk)}\right)
\end{equation*}
where $h^*\geq \frac{3}{10}\log(k)$.
\end{proposition}
This proposition gives the following corollary.
\begin{corollary}
\label{cor_lower_inner} Let $k \ge 2$.  Consider the problem instances $\{\btheta_{1},\dots,\btheta_{k}\}$ in Definition \ref{def_counterexample}.
For any algorithm, there exists an instance $\btheta \in \{\btheta_{1},\dots,\btheta_{k}\}$ and an infinite subsequence of integers $\{T_n\}_{n=1}^\infty$ such that for any $T_n$, \[\frac{200T_n}{\log (k) H\left(\btheta\right)} \ge  2\sqrt{T_n \log(6T_nk)},\]
and thus
\begin{align*}
\mathrm{R}_{\btheta} (T_n) 
&\ge \frac{1}{6} \exp\left(- \frac{200T_n}{\log(k) H\left(\btheta\right)} - 2\sqrt{T_n \log(6T_nk)}\right)
\ge \frac{1}{6} \exp\left(- \frac{400T_n}{\log(k) H\left(\btheta\right)}\right).
\end{align*}
\end{corollary}
{Corollary \ref{cor_lower_inner} directly follows from the fact that $\sqrt{T \log T} = o(T)$ and Proposition \ref{lem_lower_inner}.}

Corollary~\ref{cor_lower_inner} states that the exponent of the expected policy regret is at most $\frac{400T_n}{\log(k) H\left(\btheta\right)}$ in one of the $k$ problem instances of Definition~\ref{def_counterexample}.
Combining Lemma~\ref{lem_gamma} and Corollary~\ref{cor_lower_inner}, for any $T_n$, 
\begin{align*}
-\frac{1}{T_n}\log \mathrm{R}_{\btheta} (T_n) &\le \frac{400}{\log(k) H\left(\btheta\right)}  \le \frac{800 \Gamma^*(\btheta)}{\log(k) } , 
\end{align*}
which implies \eqref{ineq_cgammarate_kasynot} with $C=800$.
Therefore, the statement in Theorem~1~(3) does not hold.  This concludes the proof.

\section{Corrected Lemmata}
\label{sec:corrected_lemma}
Here, we correct the lemmata associated with the corrected theorem.
\begin{lemma}[Corrected Lemma~2 (2) of \citet{Kasy2021}]
\label{lem:kasy22}
Consider $\Beta(\alpha_0^d,\beta_0^d)$ priors for each $d \in \{1, \ldots, k\}$. Under any allocation rule satisfying $\frac{m^{d^{(1)}}_T}{NT}~\to~\frac{1}{2}$,
\begin{align*}
    \limsup_{T\to\infty} -\frac{1}{NT}\log\left(\sum_{d\neq d^{(1)}}p^d_T \right)\leq \Lambda^*,
\end{align*}
and under any  allocation rule satisfying $\frac{m^{d}_T}{NT} \to \lambda^{d}$ for each $d\in\{1,\dots, k\}$,
\begin{align*}
    \lim_{T\to\infty} -\frac{1}{NT}\log\left(\sum_{d\neq d^{(1)}}p^d_T \right)= \Lambda^*.
\end{align*}
\end{lemma}

\begin{lemma}[Corrected Lemma~4 of \citet{Kasy2021}. From Lemma~12 of \citet{Russo2016} and Lemma~30 of \citet{Shang2020}]
\label{lem:kasy4}
Consider any adaptive allocation rule. If 
\begin{align}
\label{eq:lem:kasy4:cond1}
\lim_{T\to\infty}\frac{1}{T} \sum_{t=1}^T q_t^{d^{(1)}} = \frac{1}{2}
\end{align}
and 
\begin{align}
\label{eq:lem:kasy4:cond2}
    \sum^\infty_{t=1}q^d_T\cdot\mathbbm{1}\left\{\frac{1}{T} \sum_{t=1}^T q_{t}^{d} > \lambda^d + \delta\right\} < \infty\quad \forall d\neq d^{(1)}, \delta > 0,
\end{align}
then 
\[\lim_{T\to\infty}\frac{1}{T} \sum_{t=1}^T q_t^{d^{(1)}} = \lambda^d.\]
\end{lemma}
\begin{lemma}[Corrected Lemma~5 of \citet{Kasy2021}. From Lemma~13 of \citet{Russo2016} and Lemma~29 of \citet{Shang2020}]
\label{lem:kasy5}
Fix any $\xi > 0$ and $d\neq d^{(1)}$. Under any allocation rule, if ${\lim_{T\to \infty} m^d_T/(NT) = 1/2}$, there exists $\xi' > 0$ and a sequence $\varepsilon_T$ with ${ \lim_{T\to \infty}\varepsilon_T =0}$ such that for any $T\in\mathbb{N}$,
\begin{align*}
    \frac{m^d_T}{NT}\geq \lambda^d + \xi \Longrightarrow \frac{p^d_T}{\max_{d\neq d^{(1)}}p^d_T} \leq \exp\left(-T(\xi' + \varepsilon_T)\right),
\end{align*}
almost surely.
\end{lemma}

\begin{lemma}[Corrected Lemma~6 of \citet{Kasy2021}.]
\label{prp_shang2020}
Denote with $\overline{D}$ the arms that are sampled only a finite amount of times:
\begin{align*}
    \overline{D} = \left\{ d\in\{1,\dots, k\}: \forall t, m^d_t < \infty \right\}.
\end{align*}
If $\overline{D}$ is empty, $p^{d^{(1)}}_t$ converges almost surely to $1$. If $\overline{D}$ is non-empty, then for every $d\in \overline{D}$, we have $\liminf_{T\to\infty} p^d_T > 0$ almost surely. 
\end{lemma}

Lemma~\ref{lem:kasy22} basically follows from Theorem~6 of \citet{Shang2020}. We prove this lemma in Appendix~\ref{appdx:lem:kasy22}. We extend the result from Beta-Bernoulli bandit model with the $\mathrm{Beta}(1, 1)$ prior to that with $\Beta(\alpha_0^d,\beta_0^d)$ priors for each $d \in \{1, \ldots, k\}$, where the constants $\alpha_0^1, \ldots, \alpha_0^k, \beta_0^1, \ldots ,\beta_0^k>0$ can be arbitrary. 

Lemmata~\ref{lem:kasy4} and \ref{lem:kasy5} are corrected citation of \citet{Russo2016} and \citet{Shang2020}.

We derive Lemma~\ref{prp_shang2020} with the help of Lemma~28 in \citet{Shang2020}.

\section{Auxiliary Results}\label{subsec_aux}
In addition to Lemmata~\ref{lem:kasy22}--\ref{prp_shang2020}, which correspond to Lemmata~2 (2), 4--6 in \citet{Kasy2021}, we additionally use the following results, Proposition~\ref{prp:lln}, \ref{prp:from_kasy}, and \ref{lem:lem_26_shang}, from \citet{Shang2020} and \citet{Kasy2021}. We further state and prove Lemma~\ref{lem:lowerbound_beta} and \ref{lem:lowerbound_XY}, which is required for the proof of Lemma~\ref{lem:kasy22}.

\begin{proposition}[From Lemma~4 of \citet{Shang2020}]
\label{prp:lln}
There exists a random variable $W$ with $\mathbb{E}\left[\exp\left( \lambda W\right)\right] < \infty$ for any $\lambda > 0$ such that 
\begin{align*}
    \ \left| m^d_T - N\sum^T_{t=1}q^d_t \right|\leq W\sqrt{(NT+1)\log(e^2 + NT)},\; \text{almost surely}, \; \forall T\in\mathbb{N}, d\in\{1,\ldots,k\}.
\end{align*}
\end{proposition}

\begin{proposition}[From Step~2 of the proof of Theorem~1 in \citet{Kasy2021}]
\label{prp:from_kasy}
Under exploration sampling, for each $d\in\{1,2,\dots, k\}$ and all $t\in\mathbb{N}$,
\begin{align*}
    \frac{p^d_t}{p^d_t + 1} \leq q^d_t \leq \frac{1}{2}. 
\end{align*}
\end{proposition}

\begin{proposition}[Lemma~26 of \cite{Shang2020}]\label{lem:lem_26_shang}
Let $X \sim \mathrm{Beta}(a_0,a_1)$ and $Y \sim \mathrm{Beta}(a_2, a_3)$ {such that} 
$$0 < \frac{a_0-1}{a_0+a_1 -1} < \frac{a_2-1}{a_2+a_3 -1}.$$  
Then, we have
$$
\mathbb{P}(X>Y) \le D \exp(-C),
$$
where 
$$
C = \inf_{\frac{a_0-1}{a_0+a_1 -1} \le y \le \frac{a_2-1}{a_2+a_3 -1}} C_{a_0,a_1}(y) + C_{a_2,a_3}(y), 
$$
$$
 C_{a_0,a_1}(y)  = (a_0 + a_1 -1)\dKL\left( \frac{a_0-1}{a_0+a_1 -1}, y\right),
$$
and
$$
D = 3 + \min \left\{ C_{a_0,a_1}\left(\frac{a_2-1}{a_2+a_3 -1}\right) ,  C_{a_2,a_3}\left(\frac{a_0-1}{a_0+a_1 -1}\right) \right\}.
$$
\end{proposition}

Proposition~\ref{prp:lln} is used for showing almost sure convergence of the shares of observations (see Remark \ref{rem_as}). Here, we note that from Proposition~\ref{prp:lln}, we can insist that with probability $1$,
\begin{align*}
    \lim_{T\to\infty}\frac{m^{d}_T}{NT} = \lambda^d\ \Longleftrightarrow\ \lim_{T\to\infty}\frac{1}{T} {\sum^T_{t=1}q^{d}_t} = \lambda^d.
\end{align*}

Besides, to prove Lemma~\ref{lem:kasy22}, we show the following lemmata.
\begin{lemma}[Lower bound on the deviation probability of Beta distribution]\label{lem:lowerbound_beta}
Let $X \sim \Beta(a_0, a_1)$. We have
\begin{align*}
    & \mathbb{P}(X \ge x) \ge \frac{\exp\left( - (a_0 + a_1 - 1) \dKL \left(\frac{a_0-1}{a_0 + a_1 -1},x\right)\right)}{a_0 + a_1},
    \\
   \text{and}\quad \quad & \mathbb{P}(X \le x) \ge \frac{\exp\left( - (a_0 + a_1 - 1) \dKL \left(\frac{a_0-1}{a_0 + a_1 -1},x\right)\right)}{a_0 + a_1}.
\end{align*}
\end{lemma}
\begin{proof}[Proof of Lemma~\ref{lem:lowerbound_beta}]
We use following facts (see e.g., Appendix I.1 of \cite{Shang2020}.). Let $F_{a_0,
a_1}^{\text{Beta}}(x)$ be the cumulative distribution function of a Beta distribution with parameters $a_0$ and $a_1$. Similarly, let $F_{a_0,a_1}^{\text{B}}(x)$ be the cumulative distribution function of a Binomial distribution with parameters $a_0$ and $a_1$. We have a following relationship (Beta-Binomial trick)
\begin{align*}
    F_{a_0,
a_1}^{\text{Beta}}(x) = 1 - F_{a_0 + a_1 - 1,
x}^{\text{B}}( a_0 -1).
\end{align*}
Therefore, we have
\begin{align*}
    \mathbb{P}(X \ge x) & = \mathbb{P}(B(a_0 + a_1-1,x) \le a_0 - 1) 
    \\
    & =  \mathbb{P}(B(a_0 + a_1-1,1 - x) \ge a_1),
\end{align*}
where $B(a_0,a_1)$ is a Binomial Distribution with parameters $a_0$ and $a_1$. From Sanov's inequality (Appendix I.1 of \cite{Shang2020}), we have
\begin{align*}
   \frac{\exp\left(  - n \dKL(x/n, p)\right)}{n+1}\le \mathbb{P}(B(n,p) \ge x)\le  \exp\left(  - n \dKL(x/n, p)\right).
\end{align*}
We get
\begin{align*}
    \mathbb{P}(X \ge x) & =  \mathbb{P}(B(a_0 + a_1-1,1 - x) \ge a_1)
    \\
    & \ge \frac{1}{a_0 + a_1} \exp\left( - (a_0 + a_1 - 1)\dKL\left( \frac{a_1}{a_0 + a_1 - 1},1-x\right) \right)
    \\
    & = \frac{1}{a_0 + a_1} \exp\left( - (a_0 + a_1 - 1)\dKL\left( \frac{a_0 - 1}{a_0 + a_1 - 1},x\right) \right),
\end{align*}
and 
\begin{align*}
    \mathbb{P}(X \le x) & = \mathbb{P}(B(a_0 + a_1-1,x) \ge a_0 - 1)
    \\
    & \ge \frac{1}{a_0 + a_1} \exp\left( - (a_0 + a_1 - 1)\dKL\left( \frac{a_0 - 1}{a_0 + a_1 - 1},x\right) \right).
\end{align*}
This concludes the proof.
\end{proof}
Then, based on Lemma~\ref{lem:lowerbound_beta}, we have the following lemma, which has an exponent that matches the exponent of Proposition~\ref{lem:lem_26_shang}.
\begin{lemma}\label{lem:lowerbound_XY}
Let $X \sim \mathrm{Beta}(a_0,a_1)$ and $Y \sim \mathrm{Beta}(a_2, a_3)$ {such that} 
$$0 < \frac{a_0-1}{a_0+a_1 -1} < \frac{a_2-1}{a_2+a_3 -1}.$$  
Then, we have
$$
\mathbb{P}(X>Y) \ge D \exp(-C),
$$
where 
$$
C = \inf_{\frac{a_0-1}{a_0+a_1 -1} \le y \le \frac{a_2-1}{a_2+a_3 -1}} C_{a_0,a_1}(y) + C_{a_2,a_3}(y), 
$$
$$
 C_{a_0,a_1}(y)  = (a_0 + a_1 -1)\dKL\left( \frac{a_0-1}{a_0+a_1 -1}, y\right),
$$
and
$$
D = \frac{1}{(a_0+ a_1)(a_2+a_3)}.
$$
\end{lemma}
\begin{proof}[Proof of Lemma~\ref{lem:lowerbound_XY}]
For each $y \in \left[\frac{a_0-1}{a_0+a_1 -1},  \frac{a_2-1}{a_2+a_3 -1}\right]$, we have
\begin{align*}
    & \mathbb{P}(X>Y) 
    \\
    & \ge \mathbb{P} \left(\{X > y\}\cap \{Y < y\}\right)
    \\
    & \ge \frac{\exp\left( - \left( (a_0 + a_1 - 1) \dKL \left(\frac{a_0-1}{a_0 + a_1 -1},y\right) + (a_2 + a_3 - 1) \dKL \left(\frac{a_2-1}{a_2 + a_3 -1},y\right)\right)\right)}{(a_0 +a_1)(a_2 + a_3)},
\end{align*}
where the last inequality is from Lemma~\ref{lem:lowerbound_beta}. 
Optimizing the right hand side of the previous inequality over $y$, we conclude the proof.
\end{proof}

\section{Proof of Theorem~\ref{thm:correct_kasy}}
\label{sec:correction}

In this section, we show the proofs of Theorem~\ref{thm:correct_kasy} (1) and (2) in Section~\ref{sec:corrected_proof} and  Theorem~\ref{thm:correct_kasy} (3) in Section~\ref{subsec:step3}. The proof procedure in Section~\ref{sec:corrected_proof} follows that of \citet{Kasy2021} as closely as possible.


\subsection{Proof of Theorem~\ref{thm:correct_kasy} (1) and (2)}
\label{sec:corrected_proof}
From Lemma~\ref{lem:kasy22}, under any allocation rule, given the event $\lim_{T\to \infty} \frac{m^{d}_T}{NT} = \lambda^{d}$, which occurs with probability $1$; we can conclude that
\begin{align*}
    -\frac{1}{NT}\log \left( \sum_{d\neq d^{(1)}}p^d_T \right) \xrightarrow{\mathrm{a.s.}} \Lambda^*.
\end{align*}
To prove Theorem~\ref{thm:correct_kasy}~(1) and (2), it suffices to show that under the exploration sampling, for each $d\in\{1,\dots,k\}$,
\begin{align}
\label{eq:allo_convergence}
    \frac{m^{d}_T}{NT} \xrightarrow{\mathrm{a.s}} \lambda^d.
\end{align}


To show that the exploration sampling satisfies \eqref{eq:allo_convergence}, we correct the following three steps in the proof of \citet{Kasy2021}, i.e., 
\begin{description}
    \item[Step~1:] Each treatment is assigned infinitely often, that is, $m^d_T \xrightarrow{\mathrm{a.s.}} \infty,\forall d \in \{1, \ldots, k\}$.
    \item[Step~2:] The share of observations $m^{d^{(1)}}_T/(NT)$ assigned to the best treatment $d^{(1)}$ converges to $\lambda^{d^{(1)}} = 1/2$ almost surely as $T\to\infty$.
    \item[Step~3:] The share of observations $m^{d}_T/(NT)$ assigned to each treatment $d\neq d^{(1)}$ converges to $\lambda^d$ almost surely as $T\to\infty$.
\end{description}
We use Step~1 to show the almost sure convergence of the posterior probability. Then, because the share of observations is determined by the posterior probability, we can show Step~2 and 3, which directly implies \eqref{eq:allo_convergence}.

\begin{proof}[Proof of Theorem~\ref{thm:correct_kasy} (1) and (2)]$\ $

\textbf{Step~1: Each treatment is assigned infinitely often.} We show $m^d_T \xrightarrow{\mathrm{a.s.}} \infty$ for each $d\in\{1,2,\dots, k\}$ using proof by contradiction.

Suppose that there exists $d'\in\{1,2,\dots, k\}$, such that $\lim_{T\to\infty} m^{d'}_T < \infty$. Under the exploration sampling, from Proposition~\ref{prp:from_kasy}, we have $q^{d'}_T \geq \frac{p^{d'}_T}{p^{d'}_T  + 1}$; therefore, by Lemma~\ref{prp_shang2020}, if $d'\in\overline{D} = \{d\in\{1,\dots, k\}: \forall T, m^d_T < \infty\}$, then $\liminf_{T\to\infty} p^{d'}_T > 0$, which implies that $\sum^\infty_{t=1}q^{d'}_T = \infty$. By Proposition~\ref{prp:lln}, we have $m^d_T \xrightarrow{\mathrm{a.s.}} \infty$. This causes a contradiction with probability $1$. Therefore, $m^d_T \xrightarrow{\mathrm{a.s.}} \infty$, for all $d\in\{1,\dots, k\}$. 

\textbf{Step~2: The share of observations $m^{d^{(1)}}_T/(NT)$ assigned to the best treatment $d^{(1)}$ converges to $1/2$ almost surely as $T\to\infty$.}
Because $m^d_T \xrightarrow{\mathrm{a.s.}} \infty$ for all $d\in\{1,\dots, k\}$, we have $p^{d^{(1)}}_T\xrightarrow{\mathrm{a.s.}}1$ from Lemma~\ref{prp_shang2020}. Then, from Proposition~\ref{prp:from_kasy}, we conclude that $ m^{d^{(1)}}_T/(NT) \xrightarrow{\mathrm{a.s.}}1/2$. 

\textbf{Step~3: The share of observations $m^{d}_T/(NT)$ assigned to each treatment $d\neq d^{(1)}$ converges to $\lambda^d$ almost surely as $T\to\infty$.} Our final step is to show \eqref{eq:allo_convergence}. From Proposition~\ref{lem:kasy4}, we can obtain this result if \eqref{eq:lem:kasy4:cond1} and \eqref{eq:lem:kasy4:cond2} hold almost surely. 

Firstly, using Proposition~\ref{prp:lln},  $m_T^{d^{(1)}}/(NT)\xrightarrow{\mathrm{a.s.}}1/2$ (the result of {\bf Step~2}) leads to
\begin{align*}
    \frac{1}{T} \sum_{t=1}^T q_t^{d^{(1)}} \xrightarrow{\mathrm{a.s.}} \frac{1}{2}.
\end{align*}
Thus, \eqref{eq:lem:kasy4:cond1} holds. 

Next, we check that \eqref{eq:lem:kasy4:cond2} holds. For $d \ne d^{(1)}$, let us define an event $\mathcal{F}^d$ as
\begin{align*}
    \mathcal{F}^d = \left\{ \lim_{T \to \infty} p_T^d = 0 \quad \text{and} \quad\forall T\in\mathbb{N},\ \left| m^d_T - N\sum^T_{t=1}q^d_t \right|\leq W\sqrt{(NT+1)\log(e^2 + NT)}\right\},
\end{align*}
where $W$ is a random variable defined in Proposition~\ref{prp:lln}.
This event $\mathcal{F}^d$ occurs with probability $1$ by Lemma~\ref{prp_shang2020} and Proposition~\ref{prp:lln}. Because each treatment is assigned infinitely often from Step~1, $\lim_{T \to \infty} p_T^d = 0$ for $d \ne d^{(1)}$ almost surely from Lemma~\ref{prp_shang2020}. The second element of $\mathcal{F}^d$ occurs with probability $1$ from Proposition~\ref{prp:lln}.

Under this event $\mathcal{F}^d$, for each constant $\xi > 0$, there exists $s$ such that for all $T \geq s$, we have
\begin{align}
\label{eq:rand_bound}
    \left| \frac{1}{T}\sum_{t=1}^T q_t^d -  \frac{m_T^d}{NT} \right| \leq \xi.
\end{align}
This is because from the second element of $\mathcal{F}^d$, for all $T$,
\[\left| \frac{m^d_T}{NT} - \frac{1}{T}\sum^T_{t=1}q^d_t \right|\leq \frac{W}{NT}\sqrt{(NT+1)\log(e^2 + NT)}\]
holds, which implies that for each each constant $\xi > 0$, \eqref{eq:rand_bound} holds for sufficiently large $T$.

 Then, the following relationship holds almost surely:
\begin{align*}
    \mathbbm{1}\left\{ \frac{1}{T}\sum_{t=1}^T q_t^d \geq \lambda^d + 2 \xi \right\} \leq \mathbbm{1}\left\{ \frac{m_T^d}{NT} \geq \lambda^d + \xi \right\},
\end{align*}
where for an event $\mathcal{E}$, $\mathbbm{1}\left\{\mathcal{E}\right\} = 1$ if the event $\mathcal{E}$ occurs.  Under the event $\mathcal{F}^d$, the exists $t_0>0$ such that for all $t\ge t_0$, $ p_t^d \le 1/2$. 
As \citet{Kasy2021} shows in Step~3 of the proof in Theorem~1, when $\max_{d\neq d^{(1)}} p^{d}_t \leq 1/2$, we have
\begin{align*}
    q^d_t \leq 2\frac{p^d_t}{\max_{d\neq d^{(1)}} p^{d}_t}.
\end{align*}

Besides, Lemma~\ref{lem:kasy5} insists that given the event ${m_T^d}/{NT} \to 1/2$,  there exists $\xi' > 0$ and a sequence $\varepsilon_T$ with $\varepsilon_T \to 0$ such that for any $T\in\mathbb{N}$, 
\[\frac{m_T^d}{NT} \geq \lambda^d + \xi \Longrightarrow \frac{p^d_T}{\max_{d\neq d^{(1)}}p^d_T} \leq \exp\left(-T(\xi' + \varepsilon_T)\right).\]
Therefore, for $d\neq d^{(1)}$, under the event $ \mathcal{F}^d$, the following inequality holds with probability $1$.
\begin{align*}
    \sum^T_{t\geq \max\{s,t_0\}}q^d_t\mathbbm{1}\left\{  \frac{1}{T}\sum_{t=1}^T q_t^d \geq \lambda^d + 2 \xi \right\} & \leq \sum^T_{t\geq \max\{s,t_0\}}q^d_t\mathbbm{1}\left\{\frac{m_T^d}{NT}  \geq \lambda^d + \xi \right\}
    \\
    & \leq \sum^T_{t\geq \max\{s,t_0\}}\exp\left(-t(\xi' + \varepsilon_t)\right) 
    \\
    & < \infty. 
\end{align*}
Therefore, \eqref{eq:lem:kasy4:cond2} holds with probability 1. 
By combining these results, from Lemma~\ref{lem:kasy4}, \eqref{eq:allo_convergence} holds. This concludes the proof.
\end{proof}

\subsection{Proof of Theorem~\ref{thm:correct_kasy} (3)}
\label{subsec:step3}
We prove Theorem~\ref{thm:correct_kasy} (3) by using Lemma~\ref{lem:kasy22}. The proof consists of two parts: derivation of the upper bound under the exploration sampling and lower bound under any allocation rule.
\begin{proof}[Proof of Theorem~\ref{thm:correct_kasy} (3)]
First, we prove $-\frac{1}{NT}\log \mathrm{W}_{\btheta}(T) \xrightarrow{\mathrm{a.s.}}  \Lambda^*$. The logarithmic posterior policy regret can be decomposed as 
\begin{align*}
    \log \mathrm{W}_{\btheta}(T) &= \log \left(\sum_{d\in\{1,\dots,k\}}\Delta^{d}\cdot p^d_T\right)\\
    &= \log \left( \sum_{d\neq d^{(1)}} \Delta^{d}\cdot p^d_T\right) \leq \log \left(\max_{d\in\{1,\ldots,k\}} \Delta^{d} \sum_{d\neq d^{(1)}}  p^d_T\right).
\end{align*}
Therefore,
\[-\frac{1}{NT}\log \mathrm{W}_{\btheta}(T) \geq -\frac{1}{NT}\log \left(\max_{d\in\{1,\ldots,k\}} \Delta^{d}\right) - \frac{1}{NT}\log \left( \sum_{d\neq d^{(1)}} p^d_T\right).\]
While $\max_{d\in\{1,\ldots,k\}} \Delta^{d}$ is constant, $\sum_{d\neq d^{(1)}} p^d_T$ decays exponentially, that is, the first converges to $0$ and the second term converges to constant, $\Lambda^*$.
Similarly, we can also show that 
\[-\frac{1}{NT}\log \mathrm{W}_{\btheta}(T) \leq -\frac{1}{NT}\log \left(\min_{d\in\{1,\ldots,k\}} \Delta^{d}\right) - \frac{1}{NT}\log \left( \sum_{d\neq d^{(1)}} p^d_T\right).\]
By taking the limit from the lower and upper bounds, we show the statement.

Next, we prove the second statement of Theorem~\ref{thm:correct_kasy} (3), which is the theoretical lower bound of algorithms.  From Lemma~\ref{lem:kasy22}, under any adaptive allocation rule satisfying $\frac{m^{d^{(1)}}_T}{NT} \xrightarrow{\mathrm{a.s.}} \frac{1}{2}$,  $\limsup_{T\to\infty}-\frac{1}{NT}\log \left( \sum_{d\neq d^{(1)}}p^d_T \right) \leq \Lambda^*$ almost surely. This implies that under any adaptive allocation rule satisfying $\frac{m^{d^{(1)}}_T}{NT} \xrightarrow{\mathrm{a.s.}} \frac{1}{2}$, $\limsup_{T\to\infty}-\frac{1}{NT}\log \mathrm{W}_{\btheta}(T) \leq \Lambda^*$ almost surely. Thus, the limit of the logarithmic posterior policy regret under the exploration sampling matches the lower bound. This concludes the proof.
\end{proof}

\section{Proof of Lemma~\ref{lem:kasy22}}
\label{appdx:lem:kasy22}
The proof follows similar steps to the proof of Theorem~6 in \cite{Shang2020}. Note that the original proof has a technical issue, and we also correct it in our proof.\footnote{In particular, transformation of Eq.~\eqref{ineq:shangfix} fixes the issue of Theorem~6 of \cite{Shang2020} by utilizing Lemma \ref{lem:lowerbound_XY}.} 
\begin{proof}
We restate the set $\overline{D}$ for the sake of readability:
\begin{align*}
    \overline{D} = \left\{ d\in\{1,\dots, k\}: \forall t, m^d_t < \infty \right\}.
\end{align*}
First, we prove the first part of the Lemma~\ref{lem:kasy22}, a lower bound on the posterior probabilities $\sum_{d\neq d^{(1)}}p^d_T$, by giving separate proofs for the cases where $D$ is the empty set or not.

\textbf{Case 1. A lower bound on the posterior probabilities when $\overline{D}$ is not empty.} The posterior variance $ \sigma_{T, d}^2 $ is computed as 
\begin{align*}
    \sigma_{T, d}^2 & = \frac{\alpha_T^d\beta_T^d}{\left(\alpha_T^d + \beta_T^d\right)^2\left(\alpha_T^d + \beta_T^d+1\right)} = \frac{\left(\alpha_0^d + r_{T-1}^d\right)\left(\beta_0^d + m_{T-1}^d-  r_{T-1}^d\right)}{\left(\alpha_0^d + \beta_0^d + m_{T-1}^d\right)^2\left(\alpha_0^d + \beta_0^d + m_{T-1}^d + 1\right)}.
\end{align*}

 Thus, when $d \in \overline{D}$, we have $\liminf_{T \to \infty}\sigma_{T, d}>0$ and $\liminf_{T \to \infty}p_T^d>0$. This means that $\limsup_{T\to \infty}p_T^{d^{(1)}} <1$. Thus, we have 
\begin{align*}
    \limsup_{T\to\infty} -\frac{1}{TN}\log\left(\sum_{d\neq d^{(1)}}p^d_T \right) &  = \limsup_{T\to\infty} -\frac{1}{TN}\log\left(1 - p^{d^{(1)}}_T \right)
    \\
    & = 0 
    \\
    & \le \Lambda^*.
\end{align*}

\textbf{Case 2. A lower bound on the posterior probabilities when $\overline{D}$ is empty.} 
When $\overline{D}$ is empty, we have
\begin{align}
    \max_{d \neq d^{(1)}} \mathbb{P}\left(\tilde{\theta}^d \ge \tilde{\theta}^{d^{(1)}} | \bm m_{T-1}, \bm r_{T-1}\right) 
    & \le 1 - p_T^{d^{(1)}} \nonumber
    \\
    &\le \sum_{d\neq d^{(1)}} \mathbb{P}\left(\tilde{\theta}^d \ge \tilde{\theta}^{d^{(1)}} | \bm m_{T-1}, \bm r_{T-1}\right) \nonumber
    \\
    & \le (k-1) \max_{d \neq d^{(1)}} \mathbb{P}\left(\tilde{\theta}^d \ge \tilde{\theta}^{d^{(1)}} | \bm m_{T-1}, \bm r_{T-1}\right) . \label{eq:equiv_Perr_max}
\end{align}
Since all the arms are sampled infinitely often, there exists $t_0$ such that for all $T \ge t_0$, for all $d \neq d^{(1)}$, 
\begin{align*}
    \frac{\alpha_0^d -1 + r_{T-1}^d}{\alpha_0^d + \beta_0^d -1 + m_{T-1}^d} 
    <
    \frac{\alpha_0^{d^{(1)}} -1 + r_{T-1}^{d^{(1)}}}{\alpha_0^{d^{(1)}} + \beta_0^{d^{(1)}} -1 + m_{T-1}^{d^{(1)}}}.
\end{align*}
For each $d \neq d^{(1)}$, define the interval, 
\begin{align*}
    I_d = \left[ \frac{\alpha_0^d -1 + r_{T-1}^d}{\alpha_0^d + \beta_0^d -1 + m_{T-1}^d}, 
    \frac{\alpha_0^{d^{(1)}} -1 + r_{T-1}^{d^{(1)}}}{\alpha_0^{d^{(1)}} + \beta_0^{d^{(1)}} -1 + m_{T-1}^{d^{(1)}}}\right].
\end{align*}
Using Proposition~\ref{lem:lem_26_shang} with $a_0 = \alpha_0^d + r_{T-1}^d$, $a_1 = \beta_0^d + m_{T -1}^d - r_{T-1}^d$, $a_2= \alpha_0^{d^{(1)}} + r_{T-1}^{d^{(1)}}$, and $a_3 = \beta_0^{d^{(1)}} + m_{T -1}^{d^{(1)}} - r_{T-1}^{d^{(1)}}$, we get 
\begin{align*}
     \mathbb{P}\left(\tilde{\theta}^d \ge \tilde{\theta}^{d^{(1)}} | \bm m_{T-1}, \bm r_{T-1}\right) & \le D \exp\left( - \inf_{y \in I_d} C_{a_0,a_1}(y) + C_{a_2,a_3}(y)\right).
\end{align*}
We have 
\begin{align*}
    D & \le 3 + (a_0 + a_1 -1 ) \dKL \left(\frac{a_0-1}{a_0+a_1-1}, \frac{a_2-1}{a_2+a_3-1}\right)
    \\
    & = 3 + \left(\alpha_0^d + \beta_0^d -1 + m_{T-1}^d\right) \dKL\left(\frac{\alpha_0^d -1 + r_{T-1}^d}{\alpha_0^d + \beta_0^d -1 + m_{T-1}^d}, \frac{\alpha_0^{d^{(1)}} -1 + r_{T-1}^{d^{(1)}}}{\alpha_0^{d^{(1)}} + \beta_0^{d^{(1)}} -1 + m_{T-1}^{d^{(1)}}}\right)
    \\
    & \le 3 + 2T N \dKL\left(0, \frac{\alpha_0^{d^{(1)}}   -1 + TN}{\alpha_0^{d^{(1)}} + \beta_0^{d^{(1)}} -1 + TN}\right)
    \\
    & \le C\left(\alpha_0^{d^{(1)}}, \beta_0^{d^{(1)}} \right) TN \log TN,
\end{align*}
with some positive constant $ C\left(\alpha_0^{d^{(1)}}, \beta_0^{d^{(1)}} \right)$. We get 
\begin{align}
    &  \limsup_{T \to \infty}\frac{1}{NT} \log \left( \frac{\mathbb{P}\left(\tilde{\theta}^d \ge \tilde{\theta}^{d^{(1)}} | \bm m_{T-1}, \bm r_{T-1}\right)}{\exp\left( - \inf_{y \in I_d} C_{a_0,a_1}(y) + C_{a_2,a_3}(y)\right)}\right) \nonumber
    \\
    & \le \limsup_{T \to \infty}\frac{1}{NT} \log \left(C\left(\alpha_0^{d^{(1)}}, \beta_0^{d^{(1)}} \right) TN \log TN\right) \nonumber
    \\
    & = 0. \label{eq:equiv_P_exp}
\end{align}

Using Lemma~\ref{lem:lowerbound_XY} with $a_0 = \alpha_0^d + r_{T-1}^d$, $a_1 = \beta_0^d + m_{T -1}^d - r_{T-1}^d$, $a_2= \alpha_0^{d^{(1)}} + r_{T-1}^{d^{(1)}}$, and $a_3 = \beta_0^{d^{(1)}} + m_{T -1}^{d^{(1)}} - r_{T-1}^{d^{(1)}}$, we get,
\begin{align*}
    \mathbb{P}\left(\tilde{\theta}^d \ge \tilde{\theta}^{d^{(1)}} | \bm m_{T-1}, \bm r_{T-1}\right) & \ge D \exp\left( - \inf_{y \in I_d} C_{a_0,a_1}(y) + C_{a_2,a_3}(y)\right),
\end{align*}
where $D = \frac{1}{\left(\alpha_0^d + \beta_0^d + m_{T -1}^d\right)\left(\alpha_0^{d^{(1)}} + \beta_0^{d^{(1)}} + m_{T -1}^{d^{(1)}}\right)}$. We get 
\begin{align}
    & \liminf_{T \to \infty} \frac{1}{NT} \log \left( \frac{\mathbb{P}\left(\tilde{\theta}^d \ge \tilde{\theta}^{d^{(1)}} | \bm m_{T-1}, \bm r_{T-1}\right)}{\exp\left( - \inf_{y \in I_d} C_{a_0,a_1}(y) + C_{a_2,a_3}(y)\right)}\right) \nonumber
    \\
    & \ge \liminf_{T \to \infty} \frac{1}{NT} \log D \nonumber
    \\
    & \ge \liminf_{T \to \infty} \frac{1}{NT} \log \frac{1}{4(NT)^2} \nonumber
    \\
    & \ge 0 \label{eq:posterior_lowerbound}
\end{align}

Therefore, for each $d \neq d^{(1)}$,  similarly to the proof of Theorem~6 of \cite{Shang2020}, we get
\begin{align}
    & 1 - p_{T}^{d^{(1)}}
    \\\nonumber
    & \underset{ \eqref{eq:equiv_Perr_max}}{\doteq} \max_{d \neq d^{(1)}} \mathbb{P}\left(\tilde{\theta}^d \ge \tilde{\theta}^{d^{(1)}} | \bm m_{T-1}, \bm r_{T-1}\right) 
    \\\nonumber
    & \underset{\eqref{eq:equiv_P_exp} \text{ and } \eqref{eq:posterior_lowerbound}}{\doteq}  \exp\Bigg( -TN \min_{d \neq d^{(1)}} \inf_{y \in I_d} \Bigg[\frac{\alpha_0^d + \beta_0^d -1 + m_{T-1}^d}{TN} 
    \dKL\left( \frac{\alpha_0^d -1 + r_{T-1}^d}{\alpha_0^d + \beta_0^d -1 + m_{T-1}^d},y\right)
    \\
    & \qquad\qquad\qquad
    + \frac{\alpha_0^{d^{(1)}} + \beta_0^{d^{(1)}} -1 + m_{T-1}^{d^{(1)}}}{TN}\dKL\left(\frac{\alpha_0^{d^{(1)}} -1 + r_{T-1}^{d^{(1)}}}{\alpha_0^{d^{(1)}} + \beta_0^{d^{(1)}} -1 + m_{T-1}^{d^{(1)}}},y \right)\Bigg]\Bigg) \label{ineq:shangfix}
    \\\nonumber
    & \ge \exp\Bigg( -TN \max_{\bm w} \min_{d \neq d^{(1)}} \inf_{y \in I_d} \Bigg[ w_d \dKL\left( \frac{\alpha_0^d -1 + r_{T-1}^d}{\alpha_0^d + \beta_0^d -1 + m_{T-1}^d},y\right)\\
    & \qquad\qquad\qquad\qquad\qquad\qquad\qquad\qquad + \frac{1}{2}\dKL\left(\frac{\alpha_0^{d^{(1)}} -1 + r_{T-1}^{d^{(1)}}}{\alpha_0^{d^{(1)}} + \beta_0^{d^{(1)}} -1 + m_{T-1}^{d^{(1)}}},y\right) \Bigg] \Bigg),\nonumber
\end{align}
where for two real-valued sequences $(a_n)$ and $(b_n)$, $a_n \doteq b_n$ denotes logarithmic equivalence, that is,
\begin{align*}
    \lim_{n\to\infty} \frac{1}{n}\log\left(\frac{a_n}{b_n}\right) = 0.
\end{align*}
For each $\varepsilon>0$, there exists $t_1(\varepsilon)>0$ such that for all $ T \ge t_1$, for all $d \neq d^{(1)}$,
\begin{align*}
    I_d \subset [\theta^d+ \varepsilon, \theta^{d^{(1)}}- \varepsilon] = I_{d, \varepsilon}^*.
\end{align*}
As the Kullback-Leibler divergence is uniformly continuous on $I_{d, \varepsilon}^*$, there exists $t_2(\varepsilon)>0$ such that for all $T \ge t_2$, 
\begin{align*}
    \dKL\left(\frac{\alpha_0^d -1 + r_{T-1}^d}{\alpha_0^d + \beta_0^d -1 + m_{T-1}^d}, y\right) \ge (1- \varepsilon) \dKL(\theta^d,y), 
\end{align*}
for all $y$ and all $d \in \{1, \ldots, k\}$. Thus, we get,
\begin{align*}
    1 - p_{T}^{d^{(1)}} &  \ge 
     \exp\left( -T N \max_{\bm w} \min_{d \neq d^{(1)}} \inf_{y \in I_{d, \varepsilon}^*}\left[ w_d \dKL( \theta^d,y) + \frac{1}{2}\dKL(\theta^{d^{(1)}},y) \right] \right)
\end{align*}
and thus,
\begin{align*}
    \limsup_{T\to \infty} -\frac{1}{NT} \log \left(\sum_{d \neq d^{(1)}} p_T^d\right) \le \Lambda^*.
\end{align*}
This concludes the proof of the first part of Lemma~\ref{lem:kasy22}.

Next, we show the second part of Lemma~\ref{lem:kasy22}. 
When  $\lim_{T \to \infty}m_{T}^d/(NT) = \lambda^d$ for all $d \in \{1, \ldots, k\}$,  we have that for each $d \in \{1, \ldots, k\}$, 
\begin{align*}
    & \lim_{T \to \infty} \inf_{y \in I_d}\Bigg[ \frac{\alpha_0^d + \beta_0^d -1 + m_{T-1}^d}{TN} \dKL\left( \frac{\alpha_0^d -1 + r_{T-1}^d}{\alpha_0^d + \beta_0^d -1 + m_{T-1}^d},y\right)\\
    &\ \ \ + \frac{\alpha_0^{d^{(1)}} + \beta_0^{d^{(1)}} -1 + m_{T-1}^{d^{(1)}}}{TN}\dKL\left(\frac{\alpha_0^{d^{(1)}} -1 + r_{T-1}^{d^{(1)}}}{\alpha_0^{d^{(1)}} + \beta_0^{d^{(1)}} -1 + m_{T-1}^{d^{(1)}}},y\right)\Bigg]
    \\
    & = \inf_{y \in \left[\theta^d, \theta^{d^{(1)}}\right]} \left[\lambda^d \dKL(\theta^d,y) + \frac{1}{2} \dKL(\theta^{d^{(1)}},y)\right]
    \\
    & = \Lambda^*.
\end{align*}
Therefore, 
\begin{align*}
    1 - p_{T}^{d^{(1)}}  & \doteq \exp\left( -T N\max_{\bm w} \min_{d \neq d^{(1)}} \inf_{y \in I_{d, \varepsilon}^*}\left[ w_d \dKL( \theta^d,y) + \frac{1}{2}\dKL(\theta^{d^{(1)}},y)\right] \right)
    \\
    & \doteq \exp(- TN \Lambda^*).
\end{align*}
Hence, we get
\begin{align*}
    \lim_{T \to \infty} - \frac{1}{TN} \log \left( \sum_{d \neq d^{(1)}} p^d_T\right) = \Lambda^*.
\end{align*}
This concludes the proof.
\end{proof}

\section{Proof of Lemma~\ref{prp_shang2020}}
\label{appdx:prp_shang2020}
The proof is analogous to the proof of Lemma~28 in \cite{Shang2020}. We extend the result from Beta-Bernoulli bandit model with the $\mathrm{Beta}(1, 1)$ prior to that with the $\Beta(\alpha_0^d,\beta_0^d)$ priors for each $d \in \{1, \ldots, k\}$, where the constants $\alpha_0^1, \ldots, \alpha_0^k, \beta_0^1, \ldots ,\beta_0^k>0$ can be arbitrary.
\begin{proof}
 When $\overline{D}$ is empty, then $\mathrm{SW}_T(d) \xrightarrow{\mathrm{a.s.}} \theta^{d}$. The posterior variance $\sigma_{T, d}^2$ is
\begin{align}\label{eq:posterior_variance}
    \sigma_{T, d}^2 & = \frac{\alpha_T^d\beta_T^d}{\left(\alpha_T^d + \beta_T^d\right)^2\left(\alpha_T^d + \beta_T^d+1\right)} = \frac{\left(\alpha_0^d + r_{T-1}^d\right)\left(\beta_0^d + m_{T-1}^d-  r_{T-1}^d\right)}{\left(\alpha_0^d + \beta_0^d + m_{T-1}^d\right)^2\left(\alpha_0^d + \beta_0^d + m_{T-1}^d + 1\right)}.
\end{align}
Therefore, under the event of $\overline{D} = \emptyset$, $\sigma_{T, d}^2 \xrightarrow{\mathrm{a.s.}} 0$ (posterior concentration).  
When $\overline{D}$ is not empty, then from \eqref{eq:posterior_variance}, we have that $\liminf_{T \to \infty}\sigma_{T, d}^2 >0$. Hence $\liminf_{T \to \infty}p_t^d > 0$. This concludes the proof.  
\end{proof}

\end{appendix}

\bibliographystyle{econometrica}
\bibliography{correction.bbl}

\end{document}